\documentclass[%
 reprint,
superscriptaddress,
 amsmath,amssymb,
pra,
floatfix,
 array,
]{revtex4-1}

\usepackage{graphicx}% Include figure files
\usepackage{dcolumn}% Align table columns on decimal point
\usepackage{bm}% bold math
\usepackage{amsthm}
\usepackage{subfigure}
\newtheorem{lemma}{Lemma}
\begin{document}

%\preprint{APS/123-QED}

\title{Quantum network coding for quantum repeaters}% Force line breaks with \\

\author{Takahiko Satoh}
 \email{satoh@is.s.u-tokyo.ac.jp}
\affiliation{
 Department of Computer Science, Graduate School of Information
 Science and Technology, The University of Tokyo, 7-3-1, Hongo,
 Bunkyo-ku, Tokyo, Japan}
%\affiliation{
% Institute for Nano Quantum Information Electronics, 4-6-1, Komaba,
% Meguro-ku, Tokyo, Japan}
% \altaffiliation[Also at ]{Graduate School of Information Science and
%   Technology, The University of Tokyo.}
\author{Fran\c cois Le Gall}
 \email{legall@is.s.u-tokyo.ac.jp}
\affiliation{
 Department of Computer Science, Graduate School of Information
 Science and Technology, The University of Tokyo, 7-3-1, Hongo,
 Bunkyo-ku, Tokyo, Japan}
\author{Hiroshi Imai}
 \email{imai@is.s.u-tokyo.ac.jp}
\affiliation{
 Department of Computer Science, Graduate School of Information
 Science and Technology, The University of Tokyo, 7-3-1, Hongo,
 Bunkyo-ku, Tokyo, Japan}
%\affiliation{
% Institute for Nano Quantum Information Electronics, 4-6-1, Komaba,
% Meguro-ku, Tokyo, Japan}
\date{\today}% It is always \today, today,

\begin{abstract}
This paper considers quantum network coding, which is a recent technique that enables
quantum information to be sent on complex networks at higher rates than by using straightforward routing strategies.
Kobayashi et al. have recently showed the potential of this technique by demonstrating how any 
classical network coding protocol gives rise to a quantum network coding protocol.
They nevertheless primarily focused on an abstract model,
in which quantum resource such as quantum registers can be freely introduced at each node.
In this work, we present a protocol for quantum network coding
%show that quantum network coding techniques can be used 
under weaker (and more practical) assumptions: our new protocol works even for 
quantum networks where adjacent nodes initially share one EPR-pair but cannot add any quantum registers or 
send any quantum information. 
A typically example of networks satisfying this assumption is {\emph{quantum repeater networks}}, 
which are promising candidates for the implementation of large scale quantum networks.
Our results thus show, for the first time, that quantum network coding techniques can increase the transmission
rate in such quantum networks as well.
\end{abstract}

\maketitle

%\tableofcontents

\section{Introduction}
Quantum communications hold potentialities which are 
qualitatively different from classical communications. 
For example, 
quantum key distribution (QKD) provides
shared, secret bits (useful for classical cryptography) whose secrecy
does not depend on the presumed difficulty of factoring large numbers or other
number-theoretic problems,
as the commonly-used Diffie-Hellman key exchange protocol does. 

Urban scale and complex topology QKD networks have already been
constructed experimentally \cite{SECOQC,Sasaki_11}.
However, it is still difficult to realize long distance quantum
communication.
Quantum repeaters \cite{repeater2,Lloyd_2004} are a potential approach
for dealing with this problem.
Quantum repeaters have three important functions:
the first is the basic physical creation of entanglement over long
distances,
the second is management of imperfections in the created quantum states (e.g.,
purification \cite{repeater,Briegel_2007,VanMeter_2009_2} or recent works
using error correction different from purification
\cite{VanMeter_2009,Munro_2010,VanMeter_2010}), 
and the third is extending entanglement from the endpoints of a single
channel to distant nodes in a topologically complex
network (e.g., entanglement swapping \cite{swapping,swapping2,VanMeter_2011}).

In quantum repeater networks, EPR-pairs are consumed as a source of
quantum communication and require a high cost for
sharing and conservation.
The communication capacity of a quantum repeater network is limited by the maximum
number of qubits the quantum repeater can store and operate on at one time.
Hence, in the future, large and complex quantum repeater networks will
be confronted with the bottleneck problem caused by shortage of quantum resources.

Meanwhile, large scale classical networks such as the Internet have
continued to increase their communication volume, and also have the bandwidth bottleneck problem.  
To address this problem, classical network coding \cite{network_coding} is drawing attention.
One of the most useful applications of this method is throughput
enhancement for certain traffic patterns: network coding is able to achieve higher throughput
than independent forwarding of every data packet, by active encoding of
the packets at intermediate nodes. 
We show an example of multiple-unicast transmission over the
directed butterfly network by using this technique in Fig.~\ref{classic}.

\begin{figure}
\includegraphics[width=86mm]{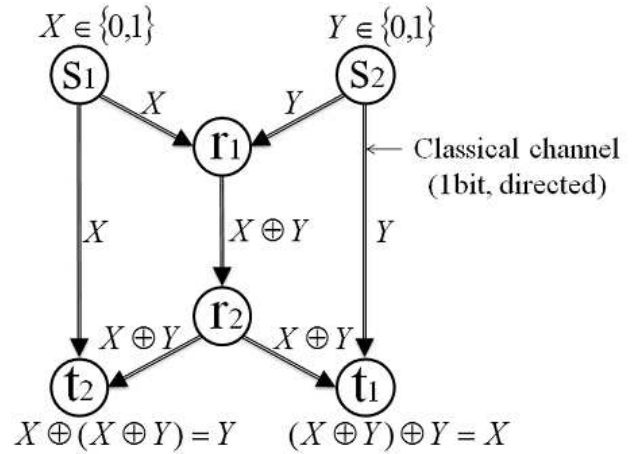}
\caption{\label{classic}The butterfly network and a classical coding
  protocol. Source nodes ($s_{1}$ and $s_{2}$) have for input bits $X$
and $Y$. The task is to send simultaneously bit $X$ from $s_1$ to $t_1$
and bit $Y$ from $s_2$ to $t_2$. This task is implemented
by using a XOR operation at relay node $r_{1}$ and target nodes
($t_{1}$ and $t_{2}$). Note this task cannot be solved by using routing.}
\end{figure}

Recently, researchers expanded network coding to include quantum
information \cite{quantum_coding,quantum_coding2,quantum_coding3,quantum_coding4} and showed that network coding using quantum information
is available without infringement of the non-cloning theorem (which forbids duplication of an 
unknown quantum state).
After that, \cite{kobayashi_coding,kobayashi_coding2,kobayashi_coding3} proved that, for any graph, quantum perfect
network coding is feasible, if free classical channels are available, whenever
classical network coding is possible. 
But these later works do not consider concrete implementation issues and,
especially, assume the availability of additional quantum resources such as quantum registers at each node of the network.
Implementation is nevertheless a fundamental problem. Indeed, in order to be able to use network coding on real quantum
networks, the amount of quantum resources required by the protocol need to be minimized.

In this work, we study quantum network coding
for practical quantum networks where adjacent 
nodes initially share one EPR-pair but cannot add any quantum registers or send any quantum information.
Typical examples are the quantum repeater networks discussed above.  
Since this setting forbids the introduction of quantum registers, 
the methods from \cite{kobayashi_coding,kobayashi_coding2,kobayashi_coding3} cannot be applied directly.
Our results nevertheless demonstrate that quantum network coding can be realized in this model as well
and are, to the best of our knowledge, the first application of network coding to
increase the transmission rate in quantum repeater networks.
This may become an effective countermeasure against communication
congestion in quantum repeater networks. 

Our results are obtained by constructing a version of the protocol in~\cite{kobayashi_coding3}
that does not require the introduction 
of any quantum register. This is non-trivial and requires new ideas.
The key idea is to convert, using only local operations and classical communication,
the EPR-pairs between adjacent nodes into appropriate entangled states
of higher dimension shared between distant nodes.
To do this, we introduce two new techniques inspired by quantum
teleportation~\cite{teleportation} and one-way quantum computation~\cite{one-way}, which we call
``Connection'' and ``Removal'', that enable us to manipulate such
entangled states and systematize the methods of encoding.

\section{Preliminaries}
\subsection{Notations}
We suppose that the reader is familiar with the basics of quantum
information theory and refer to \cite{NC} for a good reference. 
In classical information science, the fundamental unit of information is
described as a binary digit (bit).
In the case of quantum information, a quantum bit (qubit) is
the equivalent of a bit. A qubit is expressed as a superposition of two
orthonormal quantum states $\lvert 0\rangle$ and $\lvert
1\rangle$ with amplitudes $\alpha$ and $\beta$
as follows:
\begin{equation}
\lvert \psi\rangle = \alpha\lvert 0\rangle + \beta\lvert 1\rangle,
\nonumber
\end{equation}
where $\alpha$ and $\beta$ are complex number satisfying 
$\lvert \alpha \rvert^{2} + \lvert \beta \rvert^{2} = 1.$
A general quantum state of $n$ qubits can be written as
$\lvert \psi \rangle = \sum_{x \in \{0,1\}^{n}}^{} \alpha_{x}\lvert x
\rangle$, where $\alpha_{x}$ are complex numbers such that
$\sum_{x\in \{0,1\}^{n}}^{} \lvert \alpha_{x} \rvert^{2}=1$.

%\subsection{Transform operation}
In this paper, we will use the Pauli operators 
$\sigma_{X}$ and $\sigma_{Z}$ and the Hadamard operator, which are the following
single qubit transformations:
\begin{equation}
\sigma_{X}:= \lvert 1 \rangle \langle 0 \rvert + \lvert 0 \rangle
\langle 1 \rvert ,
\nonumber
\end{equation}
\begin{equation}
\sigma_{Z}:= \lvert 0 \rangle \langle 0 \rvert - \lvert 1 \rangle
\langle 1 \rvert ,
\nonumber
\end{equation}
\begin{equation}
H:= \frac{1}{\sqrt{2}}(\lvert 0 \rangle \langle 0 \rvert + \lvert 1
\rangle \langle 0 \rvert + \lvert 0 \rangle \langle 1 \rvert - \lvert
1 \rangle \langle 1 \rvert ).
\nonumber
\end{equation}

We denote by ${ \lvert + \rangle , \lvert - \rangle }$ the Hadamard basis:
\begin{equation}
\begin{split}
\lvert +\rangle = H\lvert 0\rangle =
\frac{1}{\sqrt[]{\mathstrut 2}}\left(\lvert 0 \rangle + \lvert 1
\rangle \right), \\
\lvert -\rangle = H\lvert 1\rangle =
\frac{1}{\sqrt[]{\mathstrut 2}}\left(\lvert 0 \rangle - \lvert 1
\rangle \right).
\nonumber
\end{split}
\end{equation}

We will also use the Control-NOT gate, which is the following two qubits transformation.
\begin{eqnarray}
\nonumber
{\rm CNOT^{ (\bf{A},\bf{B} )} } := \lvert 0 \rangle_{A} \lvert 0
\rangle_{B} \langle 0 \rvert_{A} \langle 0 \rvert_{B} 
&+& \lvert 0 \rangle_{A} \lvert 1
\rangle_{B} \langle 0 \rvert_{A} \langle 1 \rvert_{B} \\
\nonumber
+ \lvert 1 \rangle_{A} \lvert 1
\rangle_{B} \langle 1 \rvert_{A} \langle 0 \rvert_{B} 
&+& \lvert 1 \rangle_{A} \lvert 0
\rangle_{B} \langle 1 \rvert_{A} \langle 1 \rvert_{B} .
\end{eqnarray}

We denote by $\lvert \Psi^{+} \rangle$ and $\lvert \Phi^{+}
\rangle$ the following two qubits state (EPR-pairs):
\begin{equation}
\begin{split}
\lvert \Psi^{+} \rangle =
\frac{1}{\sqrt[]{\mathstrut 2}}\left(\lvert 00 \rangle + \lvert 11
\rangle \right), \quad
\lvert \Phi^{+} \rangle =
\frac{1}{\sqrt[]{\mathstrut 2}}\left(\lvert 01 \rangle + \lvert 10
\rangle \right)
,
\nonumber
\end{split}
\end{equation}

and by $\lvert GHZ \rangle$ the following three qubits state:
\begin{equation}
\lvert GHZ \rangle =
\frac{1}{\sqrt[]{\mathstrut 2}}\left(\lvert 000 \rangle + \lvert 111
\rangle \right).
\nonumber
\end{equation}

\subsection{Quantum repeater network}
We define a quantum repeater network as a network consisting of a number of quantum repeaters,
undirected classical channels and EPR-pairs
$\lvert\Psi^{+}\rangle$ (each pair of adjacent quantum repeaters shares one EPR-pair).
We show an example of network with three quantum
repeaters in Fig.~\ref{repeater_network}.
\begin{figure}
\includegraphics[width=76mm]{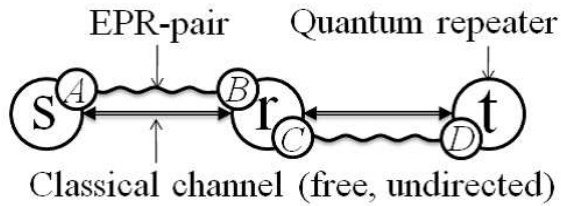}
\caption{\label{repeater_network} An example of quantum repeater
  network. There are three quantum repeaters 
  ($s$, $r$ and $t$), two free classical channels and
  two EPR-pairs $\lvert \Psi^{+}\rangle_{AB}$ and $\lvert
  \Psi^{+}\rangle_{CD}$ between $s$-$r$ and $r$-$t$, respectively.}
\end{figure}

On the network of Fig.~\ref{repeater_network}, quantum communications
are possible between adjacent repeaters ($s$-$r$ and $r$-$t$) by
teleportation using shared EPR-pairs. 
Furthermore, quantum communication between the non-adjacent repeaters
$s$ and $t$ is possible by applying entanglement swapping (the relay
repeater $r$ converts the two EPR-pairs $\lvert \Psi^{+}\rangle_{AB}\otimes\lvert
  \Psi^{+}\rangle_{CD}$ to one EPR-pair $\lvert \Psi^{+}\rangle_{AD}$).
In this way, each repeater performs quantum communication by EPR-pairs
and LOCC on this network.

The present paper will show that, for specific networks, network
coding can achieve a better throughput than this simple entanglement
swapping strategy.

\section{Quantum repeater network coding}
We present the setting for our protocol in Fig.~\ref{quantumgraph}.
\begin{figure}
\includegraphics[width=8.6cm]{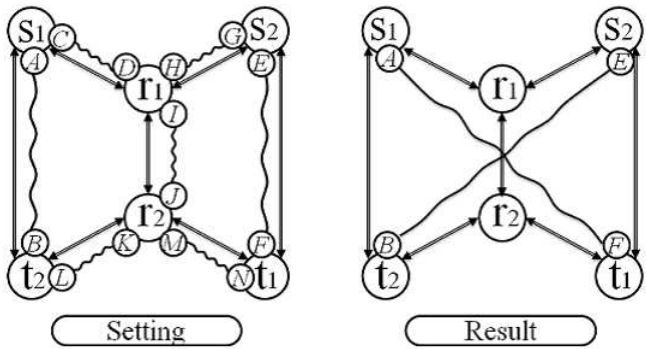}
\caption{\label{quantumgraph}The setting for our protocol. After the excution of the protocol, $s_{1}$ and $t_{1}$
  (similarly, $s_{2}$ and $t_{2}$) share one EPR-pair and
  are then able to perform quantum teleportation.}
\end{figure}
The goal of this work is to simultaneous send quantum information
between two pairs of repeaters $((s_{1},t_{1})$ and $(s_{2},t_{2}))$
located diagonally on a butterfly quantum repeaters network.
For this purpose, we construct a protocol for quantum repeater network
coding that generates EPR-pairs $\lvert\Psi^{+}\rangle_{AF}$ and
$\lvert\Psi^{+}\rangle_{BE}$ by using only LOCC and shared EPR-pairs
between adjacent quantum repeaters.
Simultaneous quantum communication between $(s_{1},t_{1})$ and
$(s_{2},t_{2})$ can then by achieved by teleportation using these EPR-pairs.

We cannot generate these EPR-pairs simultaneously by using only
entanglement swapping \cite{swapping,swapping2} because of 
the constitution of this network.
Moreover, we cannot apply existing quantum network coding methods
(\cite{kobayashi_coding3}) directly because these methods require the
introduction of intermediate quantum registers, which is not possible
in our model of quantum repeater networks.
In this work, we construct a protocol for sharing EPR-pairs without
additional quantum registers.
 
In subsection A and B we first show two new techniques.
In subsection C we give an overview of our protocol.
In subsection D we present a preliminary protocol.
In subsection E we give the final version of our protocol.

\subsection{Technique 1: Connection}
Our first technique is called Connection.
Connection is a non-unitary operation between two repeaters ($u$
and $v$, respectively).
Repeater $u$ has Control and Resource qubits ($C$ and $R$, respectively).
Repeater $v$ has a Target qubit ($T$).
We show the procedure for Connection as Table~\ref{Proc_Con}.
\begin{table}
\caption{\label{Proc_Con} ${\bf Con}^{C}_{R->T}$}
\begin{ruledtabular}
 \begin{tabular}{ll} 
& $C$ and $R$ are 1-qubit registers owned by $u$.\\
&  $T$ is a 1-qubit register owned by $v$. \\
  {\bf Step 1.}& $u$ applies $CNOT^{(C,R)}$. \\
  {\bf Step 2.}& $u$ measures $R$ in the $\{\left\lvert 0 \rangle ,\lvert 1
    \right\rangle\}$ basis.\\
          & Let $a \in \{0,1\}$ be the outcome.\\
  {\bf Step 3.}& $u$ sends $a$ to $v$ by a classical channel.\\
  {\bf Step 4.}& If $a = 1$ then $v$ applies $\sigma_{X}$ to $T$\\
 \end{tabular}
\end{ruledtabular}
\end{table}
This technique corresponds to sending one bit in the original classical
network coding scheme and is utilizing the basis manipulation method of quantum
teleportation~\cite{teleportation}.
The following lemma shows the action of Connection.

\begin{lemma}
\label{lemma_con}
Let $\lvert \Psi_{init} \rangle$ be a state of the form
\begin{equation}
\lvert \Psi_{init} \rangle = \left( \alpha \lvert \psi_{0} \rangle  \lvert 0
\rangle_{C} + \beta \lvert \psi_{1} \rangle  \lvert 1 \rangle_{C} \right) 
\otimes \lvert \Psi^{+} \rangle_{RT} \otimes \lvert \Phi \rangle,
\nonumber
\end{equation}
where $\alpha^{2}+\beta^{2}=1$ and $\lvert\psi_{0}\rangle$, $\lvert\psi_{1}\rangle$ and
$\lvert\Phi\rangle$ are arbitrary quantum states.  
Then the state after applying ${\bf Con}^{C}_{R->T}$ to $\lvert
  \Psi_{init} \rangle$ is
\begin{equation}
\begin{split}
\lvert \Psi_{final} \rangle = \left( \alpha \lvert \psi_{0} \rangle \lvert 00 \rangle_{CT}
 + \beta \lvert \psi_{1} \rangle  \lvert 11 \rangle_{CT} \right) \otimes \lvert \Phi \rangle, 
\nonumber
\end{split}
\end{equation}
where register $R$ can be disregarded.
\end{lemma}

\begin{proof}
At step 1, we apply $CNOT^{(C,R)}$. The state becomes
\begin{eqnarray}
\nonumber
\lvert \Psi_{1} \rangle = 
\alpha \lvert \psi_{0} \rangle  \lvert 0 \rangle_{C} &\otimes& \lvert \Psi^{+}
\rangle_{RT} \otimes \lvert \Phi \rangle\\
 +
\beta \lvert \psi_{1} \rangle  \lvert 1 \rangle_{C} &\otimes& \lvert \Phi^{+}
\rangle_{RT} \otimes \lvert \Phi \rangle .
\nonumber
\end{eqnarray}
At step 2, we  measure $R$. When $a = 0$ the state becomes
\begin{equation}
\begin{split}
\lvert \Psi_{2} \rangle = 
\left ( \alpha \lvert \psi_{0} \rangle  \lvert 00 \rangle_{CT} + \beta \lvert
\psi_{1} \rangle  \lvert 11 \rangle_{CT}\right )\otimes \lvert
\Phi \rangle ,
\nonumber
\end{split}
\end{equation}
and when $a = 1$ the state becomes
\begin{equation}
\begin{split}
\lvert \Psi_{2'} \rangle = 
\left ( \alpha \lvert \psi_{0} \rangle  \lvert 01 \rangle_{CT} + \beta \lvert
\psi_{1} \rangle  \lvert 10 \rangle_{CT}\right )\otimes \lvert
\Phi \rangle ,
\nonumber
\end{split}
\end{equation}
where register $R$ has been disregarded since it is not entangled anymore.
At step 4, if $a = 1$ then we apply
$\sigma_{X}$ to $T$. The state becomes 
\begin{equation}
\begin{split}
\lvert \Psi_{3} \rangle = 
\left ( \alpha \lvert \psi_{0} \rangle  \lvert 00 \rangle_{CT} + \beta \lvert
\psi_{1} \rangle  \lvert 11 \rangle_{CT}\right )\otimes \lvert
\Phi \rangle
=\lvert \Psi_{final}\rangle. \qedhere
\nonumber
\end{split}
\end{equation}
\end{proof}

For example, Lemma~\ref{lemma_con} shows that applying ${\bf
  Con}^{A}_{C->D}$ to two EPR-pairs
\begin{equation}
\lvert \Psi_{init} \rangle = \lvert\Psi^{+}\rangle_{AB} \otimes
\lvert\Psi^{+}\rangle_{CD}
\nonumber
\end{equation}
gives one GHZ-state:
\begin{equation}
\lvert \Psi_{final} \rangle = \lvert GHZ \rangle_{ABD}.
\nonumber
\end{equation}
We show the corresponding quantum circuit in Fig.~\ref{trick1}.\\
\begin{figure}
\includegraphics[width=8.6cm]{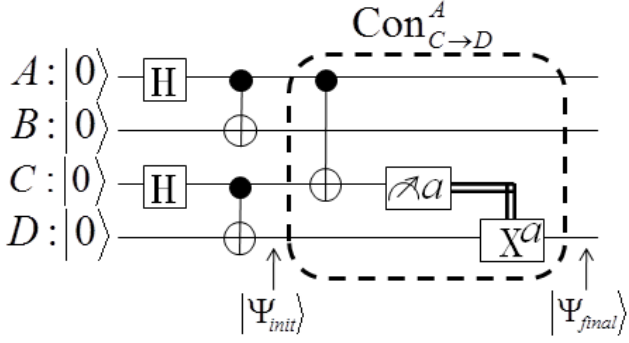}
\caption{\label{trick1}The quantum circuit for getting GHZ-states from
  two EPR-pairs by Connection.}
\end{figure}

We now show two variants of the above Connection operation.
The first variant is ``multiple resource and target qubits'' Connection and called 
  Connection:Fanout (or {\bf Fanout}). We show the procedure for Connection:Fanout
as table~\ref{Proc_Fanout}.
\begin{table}
\caption{\label{Proc_Fanout} ${\bf Fanout}^{C}_{R_{1}->T_{1},R_{2}->T_{2}}$}
\begin{ruledtabular}
 \begin{tabular}{ll} 
& $C$, $R_{1}$ and $R_{2}$ are 1-qubit registers owned by $u$. \\
&  $T_{1}$ is a  1-qubit register owned by $v$. \\
&  $T_{2}$ is a  1-qubit register owned by $w$. \\
  {\bf Step 1.}& $u$ and $v$ apply $\,\,{\bf Con}^{C}_{R_{1}->T_{1}}$. \\
  {\bf Step 2.}& $u$ and $w$ apply $\!{\bf Con}^{C}_{R_{2}->T_{2}}$. \\
 \end{tabular}
\end{ruledtabular}
\end{table}

Note that applying ${\bf Fanout}^{C}_{R_{1}->T_{1},R_{2}->T_{2}}$ is equivalent to
applying ${\bf Con}^{C}_{R_{1}->T_{1}}$ and then ${\bf
  Con}^{C}_{R_{2}->T_{2}}$. We can derive the following lemma.
\begin{lemma}
\label{lemma_copy}
Let $\lvert \Psi_{init}\rangle$ be a state of the form
\begin{equation}
\begin{split}
\lvert \Psi_{init} \rangle = \left ( \alpha \lvert\psi_{0}\rangle \lvert 0 \rangle_{A} +
\beta \lvert\psi_{1}\rangle\lvert 1 \rangle_{A} \right ) \lvert\Psi^{+}\rangle_{BC}
\lvert\Psi^{+}\rangle_{DE}\otimes\lvert\Phi\rangle,
\nonumber
\end{split}
\end{equation}
where $\alpha^{2}+\beta^{2}=1$ and $\lvert\psi_{0}\rangle$,
$\lvert\psi_{1}\rangle$ and $\lvert\Phi\rangle$ are arbitrary quantum states.  
Then the state after applying ${\bf Fanout}^{A}_{B->C,D->E}$ to $\lvert
  \Psi_{init} \rangle$ is
\begin{equation}
\nonumber
\lvert \Psi_{final} \rangle = \left( \alpha \lvert\psi_{0}\rangle\lvert 000 \rangle_{ACE}
 + \beta  \lvert\psi_{1}\rangle\lvert 111 \rangle_{ACE} \right) \otimes \lvert \Phi \rangle, 
\end{equation}
where registers $B$ and $D$ can be disregarded.
\end{lemma}
\begin{proof}
At step 1, we apply ${\bf Con}^{A}_{B->C}$ to
$\lvert\Psi_{init}\rangle$. From Lemma~\ref{lemma_con}, the state
becomes
\begin{equation}
\begin{split}
\lvert \Psi_{1} \rangle = \left ( \alpha \lvert\psi_{0}\rangle \lvert 00 \rangle_{AC} +
\beta \lvert\psi_{1}\rangle \lvert 11 \rangle_{AC} \right ) 
\lvert\Psi^{+}\rangle_{DE}\otimes\lvert\Phi\rangle,
\nonumber
\end{split}
\end{equation}
where register $B$ has been disregarded since it is not entangled anymore.
At step 2, we apply ${\bf Con}^{A}_{D->E}$ to
$\lvert\Psi_{1}\rangle$. From Lemma~\ref{lemma_con}, the state
becomes
\begin{eqnarray}
\nonumber
\lvert \Psi_{2} \rangle & = & \left( \alpha \lvert\psi_{0}\rangle \lvert 000 \rangle_{ACE} +
\beta \lvert\psi_{1}\rangle \lvert 111 \rangle_{ACE} \right)\otimes
\lvert\Phi\rangle \\
 & = &\lvert\Psi_{final} \rangle,
\nonumber
\end{eqnarray}
where register $D$ has been disregarded since it is not entangled anymore.
\qedhere
\end{proof}

For example, Lemma~\ref{lemma_copy} shows that applying ${\bf Fanout}^{A}_{B->C,D->E}$ to 
\begin{equation*}
\begin{split}
\lvert \Psi_{init} \rangle = \left ( \alpha \lvert 0 \rangle_{A} +
\beta \lvert 1 \rangle_{A} \right ) \lvert\Psi^{+}\rangle_{BC}
\lvert\Psi^{+}\rangle_{DE}
\nonumber
\end{split}
\end{equation*}
gives the following quantum state:
\begin{equation*}
\begin{split}
\lvert \Psi_{final} \rangle =\alpha \lvert 000 \rangle_{ACE} + \beta \lvert 111
\rangle_{ACE}.
\end{split}
\end{equation*}
We show the corresponding circuit in Fig.~\ref{con_cop}.
\begin{figure}
\includegraphics[width=8.6cm]{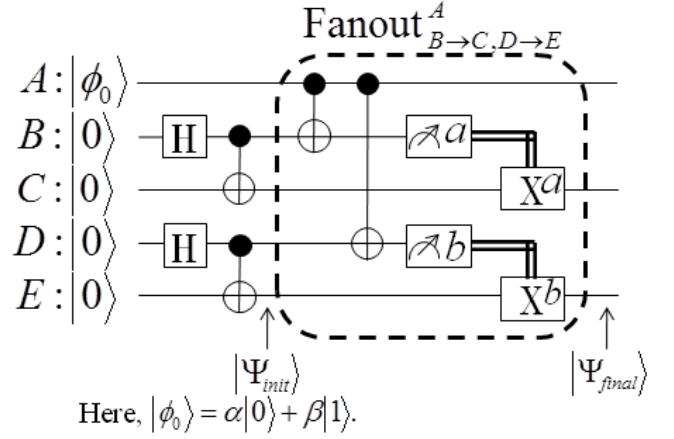}
\caption{\label{con_cop}The circuit for Connection:Fanout.}
\end{figure}

The next variant is ``multiple control qubits'' Connection and called 
  Connection:Add (or {\bf Add}). We show the procedure for Connection:Add
as Table~\ref{Proc_Add},
\begin{table}
\caption{\label{Proc_Add} ${\bf Add}^{C_{1},C_{2}}_{R->T}$}
\begin{ruledtabular}
 \begin{tabular}{ll} 
& $C_{1}$, $C_{2}$ and $R$ are 1-qubit registers owned by $u$.\\
&  $T$ is a  1-qubit registers owned by $v$. \\
  {\bf Step 1.}& $u$ applies $CNOT^{(C_{1},R)}$. \\
  {\bf Step 2.}& $u$ and $v$ apply $Con^{C_{2}}_{R->T}$. \\
 \end{tabular}
\end{ruledtabular}
\end{table}
and prove the following lemma.
\begin{lemma}
\label{lemma_add}
Let $\lvert \Psi_{init} \rangle$ be a state of the form
\begin{eqnarray}
\nonumber
\lvert \Psi_{init} \rangle =
(\alpha \lvert\psi_{0}\rangle \lvert 0 \rangle_{A} & + &
\beta \lvert\psi_{1}\rangle\lvert 1 \rangle_{A} ) \\
\otimes ( \gamma \lvert\phi_{0}\rangle \lvert 0\rangle_{B}  
 & + & \delta \lvert\phi_{1} \rangle\lvert 1 \rangle_{B} )
 \lvert\Psi^{+}\rangle_{CD}\otimes\lvert\Phi\rangle ,
\nonumber
\end{eqnarray}
where $\alpha^{2}+\beta^{2}=\gamma^{2}+\delta^{2}=1$ and
$\lvert\psi_{0}\rangle$, $\lvert\psi_{1}\rangle$, $\lvert\phi_{0}\rangle$,
$\lvert\phi_{1}\rangle$ and $\lvert\Phi\rangle$ are arbitrary quantum states.
Then the state after applying ${\bf Add}^{A,B}_{C->D}$ to $\lvert
  \Psi_{init} \rangle$ is
\begin{equation}
\begin{split}
\nonumber
\lvert \Psi_{final} \rangle =
 \big(\left(\alpha \gamma
 \lvert\psi_{0}\rangle\lvert\phi_{0}\rangle\lvert 00 \rangle_{AB} + 
\beta \delta \lvert\psi_{1}\rangle\lvert\phi_{1}\rangle\lvert 11
\rangle_{AB} \right)\lvert 0 \rangle_{D} \\ 
\nonumber +
( \alpha \delta \lvert\psi_{0}\rangle\lvert\phi_{1}\rangle\lvert 01 \rangle_{AB} +
\beta \gamma \lvert\psi_{1}\rangle\lvert\phi_{0}\rangle\lvert 10
\rangle_{AB} )\lvert 1 \rangle_{D}\big)\lvert\Phi\rangle.
\nonumber
\end{split}
\end{equation}
where register $C$ can be disregarded.
\end{lemma}
\begin{proof}
At step 1, we apply $CNOT^{(A,C)}$. The state becomes
\begin{eqnarray}
\nonumber
\lvert \Psi_{1} \rangle &=& 
\big(\alpha  \lvert\psi_{0}\rangle\lvert 0 \rangle_{A} \left( \gamma
\lvert\phi_{0}\rangle\lvert 0\rangle_{B} 
 +\delta \lvert\phi_{1}\rangle\lvert 1 \rangle_{B} \right) \otimes \lvert \Psi^{+}
\rangle_{CD} \\
 &+&
\beta \lvert\psi_{1}\rangle\lvert 1 \rangle_{A} \left( \gamma
\lvert\phi_{0}\rangle\lvert 0\rangle_{B} 
 +\delta \lvert\phi_{1}\rangle\lvert 1 \rangle_{B} \right) \otimes \lvert \Phi^{+}
\rangle_{CD}\big)\lvert\Phi\rangle.
\nonumber
\end{eqnarray}
From Lemma~\ref{lemma_con}, the final state:
\begin{eqnarray*}
\lvert \Psi_{4} \rangle &=& \big(\left(\alpha \gamma
 \lvert\psi_{0}\rangle\lvert\phi_{0}\rangle\lvert 00 \rangle_{AB} + 
\beta \delta \lvert\psi_{1}\rangle\lvert\phi_{1}\rangle\lvert 11
\rangle_{AB} \right)\lvert 0 \rangle_{D} \\ 
&+& ( \alpha \delta \lvert\psi_{0}\rangle\lvert\phi_{1}\rangle\lvert 01 \rangle_{AB} +
\beta \gamma \lvert\psi_{1}\rangle\lvert\phi_{0}\rangle\lvert 10
\rangle_{AB} )\lvert 1 \rangle_{D}\big)\lvert\Phi\rangle\\
 &=&\lvert \Psi_{final}\rangle. 
\end{eqnarray*}
where register $C$ has been disregarded since it is not entangled anymore.
\qedhere
\end{proof}

For example, Lemma~\ref{lemma_add} shows that applying ${\bf Add}^{A,B}_{C->D}$ to
\begin{equation}
\begin{split}
\lvert \Psi_{init} \rangle= \left( \alpha \lvert 0 \rangle_{A} +
\beta \lvert 1 \rangle_{A} \right) \otimes \left( \gamma \lvert 0 \rangle_{B} +
\delta \lvert 1 \rangle_{B} \right) \lvert\Psi^{+}\rangle_{CD}
\nonumber
\end{split}
\end{equation}
gives the following quantum state: 
\begin{eqnarray}
\nonumber
\lvert \Psi_{final} \rangle= \left( \alpha \gamma \lvert 00 \rangle_{AB} +
\beta \delta \lvert 11 \rangle_{AB} \right)&\lvert& 0 \rangle_{D} \\
+\left( \alpha \delta \lvert 01 \rangle_{AB} +
\beta \gamma \lvert 10 \rangle_{AB} \right)&\lvert& 1 \rangle_{D}.
\nonumber
\end{eqnarray}
We show the corresponding circuit in Fig.~\ref{con_add}.
\begin{figure}
\includegraphics[width=8.6cm]{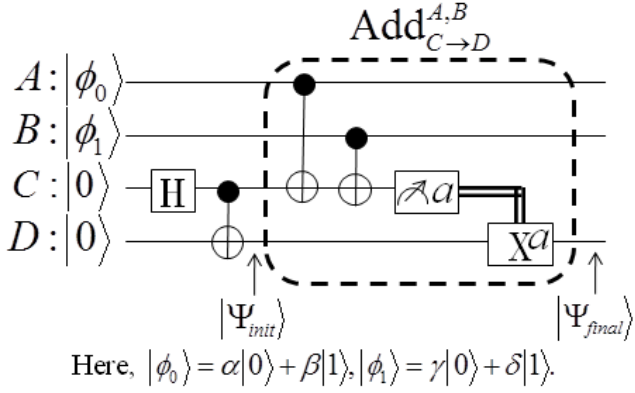}
\caption{\label{con_add}The circuit for Connection:Add.}
\end{figure}

\subsection{Technique 2: Removal}
Our second technique is called Removal.
Removal is a non-unitary operation between two repeaters ($u$ and $v$,
respectively) which deletes a resource qubit 
($R$) of a quantum state
using measurement in the Hadamard basis and
$\sigma_{Z}$.
The procedure for Removal is shown as Table~\ref{Proc_Rem}.
\begin{table}
\caption{\label{Proc_Rem} ${\bf Rem}_{R->T}$}
\begin{ruledtabular}
 \begin{tabular}{ll}
& $R$ is a 1-qubit register owned by $u$. \\
& $T$ is a 1-qubit register owned by $v$. \\
  {\bf Step 1.}& $u$ applies the Hadamard gate to $R$. \\
  {\bf Step 2.}& $u$ measures $R$ in $\{ \left\lvert 0 \rangle ,\lvert 1 \right\rangle\}$ basis.\\
          & Let $a \in \{0,1\}$ be the outcome.\\
  {\bf Step 3.}& $v$ sends $a$ to $v$ by classical channel. \\
  {\bf Step 4.}& If $a = 1$ then repeater $v$ applies $\sigma_{Z}$ to $T$\\
 \end{tabular}
\end{ruledtabular}
\end{table}
This technique is inspired by the qubit removal method using pauli measurements in
one-way quantum computing~\cite{one-way} (e.g., qubit removal from the graph states
by using a $Z$ basis) measurement).
The following lemma shows the action of Removal.
\begin{lemma}
\label{lemma_rem}
Let $\lvert \Psi_{init} \rangle$ be a state of the form
\begin{equation}
\begin{split}
\lvert \Psi_{init} \rangle = \left(\alpha \lvert 00 \rangle_{AB}  \lvert \psi_{00}
\rangle + \beta \lvert 11 \rangle_{AB}  \lvert
\psi_{11} \rangle\right)\otimes\lvert\Phi\rangle,
\nonumber
\end{split}
\end{equation}
where $\lvert \alpha \rvert^{2} + \lvert \beta \rvert^{2} =1$, and
$\lvert\psi_{00}\rangle$, $\lvert\psi_{11}\rangle$,
$\lvert\Phi\rangle$ are arbitrary quantum states. 
Then by applying ${\bf Rem}_{A->B}$ on $\lvert \Psi_{init} \rangle$, we
obtain the state
\begin{equation}
\begin{split}
\lvert \Psi_{final} \rangle = \left(\alpha \lvert 0 \rangle_{B}  \lvert \psi_{00}
\rangle + \beta \lvert 1 \rangle_{B}  \lvert
\psi_{11} \rangle\right)\otimes\lvert\Phi\rangle.
\nonumber
\end{split}
\end{equation}
where register C can be disregarded.
\end{lemma}
\begin{proof}
At step 1, we apply the Hadamard gate to $A$. The state becomes
\begin{equation}
\lvert \Psi_{1} \rangle = \left(\alpha \lvert +0 \rangle_{AB}  \lvert \psi_{00}
\rangle + \beta \lvert -1 \rangle_{A,B}  \lvert
\psi_{11} \rangle\right)\otimes\lvert\Phi\rangle.
\nonumber
\end{equation}
At step 2, we measure $A$. After this step, when $a = 0$ the state
is 
\begin{equation}
\lvert \Psi_{2} \rangle = \left(\alpha \lvert 0 \rangle_{B}  \lvert \psi_{00}
\rangle + \beta \lvert 1 \rangle_{B}  \lvert
\psi_{11} \rangle\right)\otimes\lvert\Phi\rangle,
\nonumber
\end{equation}
and when $a = 1$ the state is 
\begin{equation}
\lvert \Psi_{2'} \rangle = \left(\alpha \lvert 0 \rangle_{B}  \lvert \psi_{00}
\rangle - \beta \lvert 1 \rangle_{B}  \lvert
\psi_{11} \rangle\right)\otimes\lvert\Phi\rangle ,
\nonumber
\end{equation}
where register $A$ can be disregarded since it is not entangled
anymore. At step 4, if $a = 1$ then we apply $\sigma_{Z}$ to
$T$. The state becomes
\begin{equation}
\begin{split}
\lvert \Psi_{4} \rangle = \left(\alpha \lvert 0 \rangle_{B}  \lvert \psi_{00}
\rangle + \beta \lvert 1 \rangle_{B}  \lvert
\psi_{11} \rangle\right)\otimes\lvert\Phi\rangle =\lvert \Psi_{final} \rangle.\qedhere
\nonumber
\end{split}
\end{equation}
\end{proof}
Lemma~\ref{lemma_rem} shows that Removal is able to ``delete'' the
target qubit used in a Connection operation (compare with Lemma~\ref{lemma_con}).
For instance, by applying ${\bf Rem}_{A->B}$ on the GHZ-state
\begin{equation}
\lvert \Psi_{init} \rangle = \lvert GHZ\rangle_{ABC},
\nonumber
\end{equation}
 we obtain the EPR-pair
\begin{equation}
\nonumber
\lvert \Psi_{final} \rangle = \lvert\Psi^{+}\rangle_{BC}.
\end{equation}
The corresponding circuit is shown in Fig.~\ref{trick2}.\\
\begin{figure}
\includegraphics[width=8.6cm]{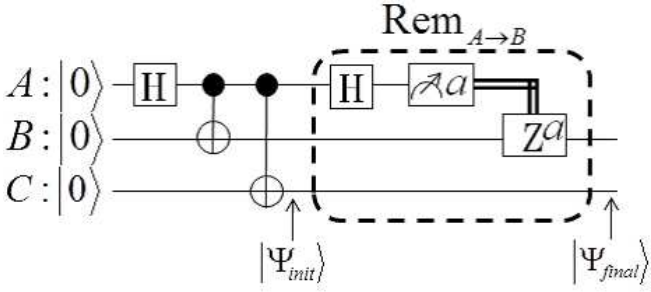}
\caption{\label{trick2}The circuit for Removal.}
\end{figure}

We now present a variant of {\bf Rem} that will delete the target qubit used in Connection:Add
operation.
We call this operation Removal:Add ({\bf RemAdd}) and show the
procedure as Table~\ref{Proc_RemAdd}.
We can derive the following lemma.
\begin{table}
\caption{\label{Proc_RemAdd} ${\bf RemAdd}_{R->T_{1},T_{2}}$}
\begin{ruledtabular}
 \begin{tabular}{ll} 
& $R$ is a 1-qubit register owned by $u$. \\
& $T_{1}$ is a 1-qubit register owned by $v$.\\ 
& $T_{2}$ is a 1-qubit register owned by $w$.\\ 
  {\bf Step 1.}& Repeater $u$ applies the Hadamard gate to $R$. \\
  {\bf Step 2.}& $u$ measures $R$ in $\{ \left\lvert 0 \rangle ,\lvert 1 \right\rangle\}$ basis.\\
          & Let $a$ be the outcome.\\
  {\bf Step 3.}& $v$ sends $a \in \{0,1\}$ to $v$ and $w$ by classical channel. \\
  {\bf Step 4.}& If $a = 1$ then repeater $v$ and $w$ applies
  $\sigma_{Z}$ \\
 & to $T_{1}$ and $T_{2}$\\
 \end{tabular}
\end{ruledtabular}
\end{table}
\begin{lemma} 
\label{lemma_remadd}
Let $\lvert \Psi_{init} \rangle$ be a state of the form
\begin{equation}
\begin{split}
\lvert \Psi_{init} \rangle = \left(\sum^{1}_{i,j=0} a_{ij} \lvert ij \rangle_{AB}
\lvert i \oplus j \rangle_{C}  \lvert \psi_{ij} \rangle\right)\otimes\lvert\Phi\rangle,
\nonumber
\end{split}
\end{equation}
where $\sum_{i,j}^{}\lvert a_{ij} \rvert^{2} =1$, and
$\lvert\psi_{i,j}\rangle$, $\lvert\Phi\rangle$ are arbitrary quantum
states. Then by applying
${\bf RemAdd}_{C->A,B}$, we obtain the state
\begin{equation}
\begin{split}
\lvert \Psi_{final} \rangle = \left(\sum_{i,j=0}^{1} a_{ij} \lvert i j \rangle_{AB}  \lvert \psi_{ij}
\rangle\right)\otimes\lvert\Phi\rangle,
\nonumber
\end{split}
\end{equation}
where register $C$ can be disregarded.
\end{lemma}
\begin{proof}
At step 1, we apply the Hadamard gate to $C$. The state becomes
\begin{eqnarray}
\nonumber
\lvert \Psi_{1} \rangle = \big(\lvert +\rangle_{C}\otimes(a_{00} \lvert 00 \rangle_{AB}
\lvert \psi_{00}  \rangle + a_{11} \lvert 11 \rangle_{AB} \lvert
\psi_{11} &\rangle&) \\ + \lvert -\rangle_{C}\otimes(a_{01} \lvert 01 \rangle_{AB}
\lvert \psi_{01} \rangle + a_{10} \lvert 10 \rangle_{AB} \lvert
\psi_{10} \rangle)& \big)&\otimes\lvert\Phi\rangle.
\nonumber
\end{eqnarray}
At step 2, we measure $C$. When $a =0$ the state becomes
\begin{eqnarray}
\nonumber
\lvert \Psi_{2} \rangle = \big((a_{00} \lvert 00 \rangle_{AB}
\lvert \psi_{00}  \rangle &+& a_{11} \lvert 11 \rangle_{AB} \lvert
\psi_{11} \rangle) \quad\quad\quad\\ + (a_{01} \lvert 01 \rangle_{AB}
\lvert \psi_{01} \rangle &+& a_{10} \lvert 10 \rangle_{AB} \lvert
\psi_{10} \rangle) \big)\otimes\lvert\Phi\rangle,
\nonumber
\end{eqnarray}
and when $a =1$ the state becomes
\begin{eqnarray}
\nonumber
\lvert \Psi_{2'} \rangle = \big((a_{00} \lvert 00 \rangle_{AB}
\lvert \psi_{00}  \rangle &+& a_{11} \lvert 11 \rangle_{AB} \lvert
\psi_{11} \rangle) \quad\quad\quad\\ - (a_{01} \lvert 01 \rangle_{AB}
\lvert \psi_{01} \rangle &+& a_{10} \lvert 10 \rangle_{AB} \lvert
\psi_{10} \rangle) \big)\otimes\lvert\Phi\rangle,
\nonumber
\end{eqnarray}
where register $C$ can be disregarded since it is not entangled
anymore. At step 4, if $a =1$ then we apply $\sigma_{Z}$ to
$A$ and $B$. The state becomes
\begin{equation}
\begin{split}
\lvert \Psi_{4} \rangle = \left(\sum a_{ij} \lvert i j \rangle_{AB}  \lvert \psi_{ij}
\rangle \right)\otimes\lvert\Phi\rangle =\lvert \Psi_{final} \rangle. \qedhere
\nonumber
\end{split}
\end{equation}
\end{proof}
For example, Lemma~\ref{lemma_remadd} shows that applying ${\bf RemAdd}_{C->A,B}$ to
\begin{equation}
\begin{split}
\lvert \Psi_{init} \rangle = \left( \alpha \gamma \lvert 00 \rangle_{AB} +
\beta \delta \lvert 11 \rangle_{AB} \right)\lvert 0 \rangle_{C} \\
+\left( \alpha \delta \lvert 01 \rangle_{AB} +
\beta \gamma \lvert 10 \rangle_{AB} \right)\lvert 1 \rangle_{C}
\nonumber
\end{split}
\end{equation}
gives the following quantum state:
\begin{equation}
\begin{split}
\lvert \Psi_{final} \rangle = \left( \alpha \lvert 0 \rangle_{A} +
\beta \lvert 1 \rangle_{A} \right)\otimes 
\left( \gamma \lvert 0 \rangle_{B} +
\delta \lvert 1 \rangle_{B} \right).
\nonumber
\end{split}
\end{equation}
The corresponding circuit is shown in Fig.~\ref{rem_add}.
\begin{figure}
\includegraphics[width=8.6cm]{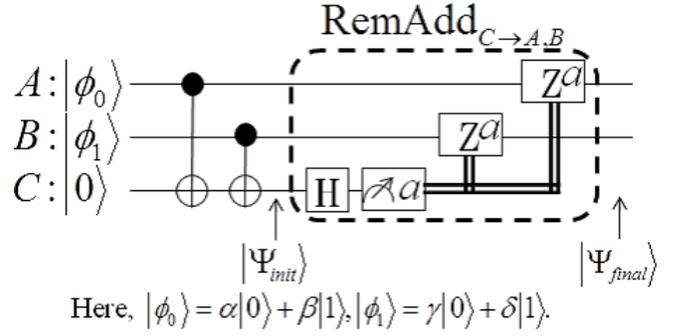}
\caption{\label{rem_add}The circuit for Removal:Add.}
\end{figure}

\subsection{Overview of our encoding protocol}
We now give an overview of our protocol for the butterfly quantum
repeater network of Fig.~\ref{quantumgraph}.
We will give a complete description of our protocol in subsections D
and E.

The first half of our protocol simulates the classical strategy of
Fig.~\ref{classic}. For this
purpose, each repeater applies {\bf Add} or {\bf Fanout} operations. We
show the correspondence between classical and quantum operations in
Fig.~\ref{cor_all}. Applying ${\bf Add}^{B,D}_{E->F}$ to
\begin{equation}
\lvert
\Psi_{init}\rangle=\lvert\Psi^{+}\rangle_{AB}\otimes\lvert\Psi^{+}\rangle_{CD}\otimes\lvert\Psi^{+}\rangle_{EF}
\nonumber
\end{equation}
  gives the following quantum state:
\begin{eqnarray*}
\lvert\Psi_{final}\rangle= \frac{1}{2}(\lvert 0000\rangle_{ABCD}+\lvert
  1111\rangle_{ABCD})&\otimes&\lvert 0\rangle_{F}\\
+\frac{1}{2}(\lvert
  1100\rangle_{ABCD}+\lvert 1100\rangle_{ABCD})&\otimes&\lvert
  1\rangle_{F}.
\nonumber
\end{eqnarray*}
Thus ${\bf Add}^{B,D}_{E->F}$ corresponds to
the classical parity operation (computing the parity of $B$ and $D$ into
register $F$). Applying ${\bf Fanout}^{B}_{C->D, E->F}$ to
\begin{equation}
\lvert \Psi_{init}\rangle=  \lvert\Psi^{+}\rangle_{AB}\otimes\lvert\Psi^{+}\rangle_{CD}\otimes\lvert\Psi^{+}\rangle_{EF}
\nonumber
\end{equation}
gives the following quantum state:
\begin{equation}
\lvert \Psi_{final}\rangle=   \frac{1}{\sqrt{2}}(\lvert 0000\rangle_{ABDF}+\lvert 1111\rangle_{ABDF}).
\nonumber
\end{equation}
${\bf Fanout}^{B}_{C->D, E->F}$ corresponds to the classical fanout
operation (copying $B$ into registers $D$ and $F$).\\
\begin{figure}[tb]
\begin{minipage}{86mm}
\subfigure[ \; Classical parity and quantum add operation.]
{\includegraphics[width=85mm]{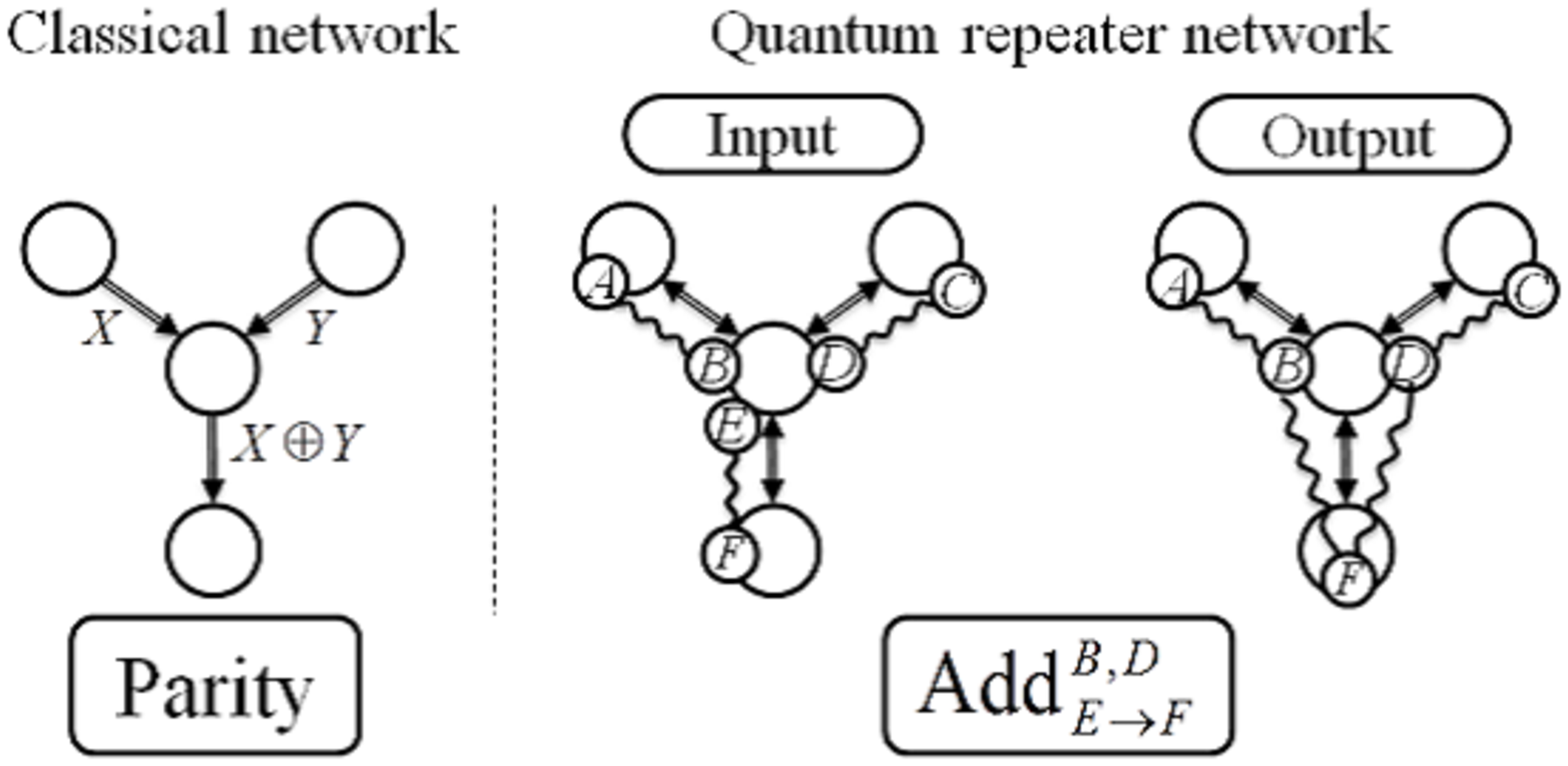}}
\end{minipage}
\begin{minipage}{86mm}
\subfigure[ \; Classical and quantum fanout operation.]
{\includegraphics[width=85mm]{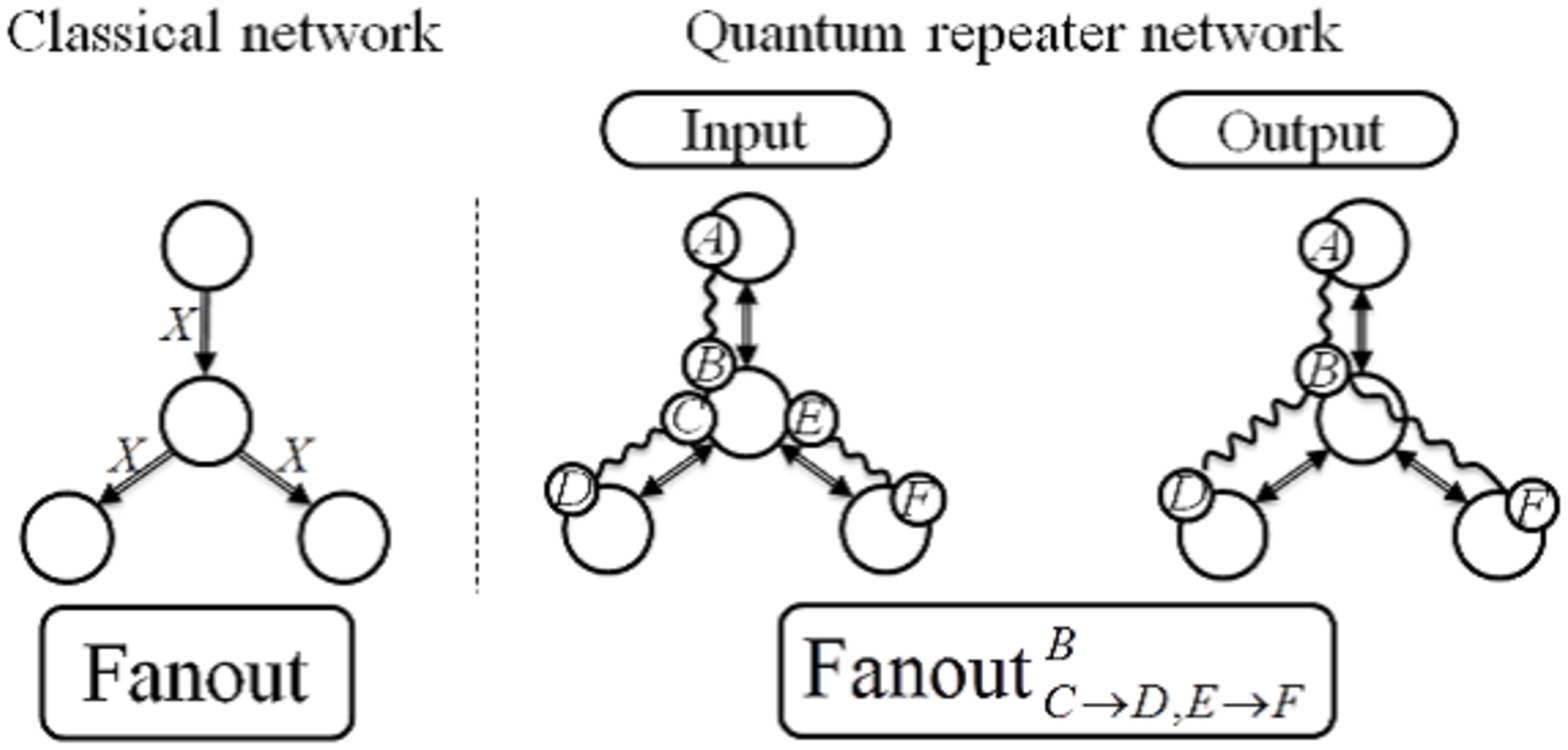}}
\end{minipage}
\caption{\label{cor_all} The correspondence between classical and
  quantum coding.}
\end{figure}

The second half of our protocol will delete redundant registers by {\bf Rem}
and {\bf RemAdd} operations.

A difficulty is that this correspondence of Fig.~\ref{cor_all} cannot be used for the
encoding performed at the sender nodes because the sender nodes have
no control qubit (i.e., the sender nodes have no incoming edges).
To deal with this, we will introduce additional
registers as control qubits and describe quantum repeater network coding
protocol for this setting in subsection D.
We will then show that our protocol can work without
additional registers in subsection E.

\subsection{Encoding with additional registers}
In this subsection we assume that $s_{1}$ has an additional register
$A'$ with state $\lvert + \rangle_{A'}$ and $s_{2}$ has an additional register
$E'$ with state $\lvert + \rangle_{E'}$. The registers $A'$ and $E'$
will be used as control qubits for the {\bf Fanout} operations performed by
the senders.
We present our procedure as Table~\ref{Proc_Enc}, and show
below the evolution of the quantum state of the system.
\begin{table}
\caption{\label{Proc_Enc} Encoding}
\begin{ruledtabular}
 \begin{tabular}{ll}
  {\bf Step 1.}& $s_{1}$ and $r_{1}$ apply ${\bf Fanout}^{A'}_{A->B,
    C->D}$,\\
          & $s_{2}$ and $r_{1}$ apply ${\bf Fanout}^{E'}_{E->F,G->H}$.  \\
  {\bf Step 2.}& $r_{1}$ and $r_{2}$ apply ${\bf Add}^{{D,H}}_{I->J}$. \\
  {\bf Step 3.}& $r_{2}$, $t_{1}$ and $t_{2}$ apply ${\bf Fanout}^{J}_{K->L, M->N}$. \\
  {\bf Step 4.}& $t_{1}$ applies ${\bf CNOT}^{(N,F)}$, $t_{2}$ applies ${\bf CNOT}^{(L,B)}$.\\
  {\bf Step 5.}& $t_{2}$ and $r_{2}$ apply ${\bf Rem}_{L->J}$,\\
          & $t_{1}$ and $r_{2}$ apply ${\bf Rem}_{N->J}$. \\
  {\bf Step 6.}& $r_{2}$ and $r_{1}$ apply ${\bf RemAdd}_{J->{D,H}}$. \\
  {\bf Step 7.}& $r_{1}$ and $s_{1}$ apply ${\bf Rem}_{D->A'}$,\\
          & $r_{1}$ and $s_{2}$ apply ${\bf Rem}_{H->E'}$. \\
 \end{tabular}
\end{ruledtabular}
\end{table}

The input state is
\begin{equation}
\label{inputstate}
\begin{split}
\lvert \Psi_{0} \rangle=
\lvert +\rangle_{A'}\lvert\Psi^{+}\rangle_{AB}\lvert\Psi^{+}\rangle_{CD}
\lvert +\rangle_{E'}\lvert\Psi^{+}\rangle_{EF}\lvert\Psi^{+}\rangle_{GH}\\
\otimes \lvert\Psi^{+}\rangle_{IJ}
\lvert\Psi^{+}\rangle_{KL}\lvert\Psi^{+}\rangle_{MN}.
\end{split}
\end{equation}

At step 1, $s_{1}$ and $t_{2}$ apply Connection:Fanout. ($s_{2}$ and
$t_{1}$ do the same.)
From Lemma \ref{lemma_copy}, the state becomes 
\begin{equation*}
\lvert \Psi_{1} \rangle=
\lvert GHZ \rangle_{A'BD}
\lvert GHZ \rangle_{E'FH}
\lvert \Psi^{+} \rangle_{IJ}
\lvert \Psi^{+} \rangle_{KL}
\lvert \Psi^{+} \rangle_{MN}.
\end{equation*}

At step 2, $r_{1}$ and $r_{2}$ apply Connection:Add.
From Lemma~\ref{lemma_add}, the state becomes
\begin{equation*}
\begin{split}
\left|\Psi_{2}\right\rangle=
\frac{1}{2}\left(\left\lvert 000000\right\rangle_{A'BDE'FH}+\lvert
111111 \rangle_{A'BDE'FH} \right)\lvert 0 \rangle_{J} \\
\otimes \lvert \Psi^{+} \rangle_{KL}\lvert \Psi^{+} \rangle_{MN}\\
+\frac{1}{2}\left(\left|000111\right\rangle_{A'BDE'FH}+\left|111000\right\rangle_{A'BDE'FH}\right)
\lvert 1 \rangle_{j}\\
\otimes  \lvert \Psi^{+} \rangle_{KL}\lvert \Psi^{+} \rangle_{MN}.
\end{split}
\end{equation*}

At step 3, $r_{2}$, $t_{1}$ and $t_{2}$ apply Connection:Fanout.
From Lemma~\ref{lemma_copy}, The state becomes
\begin{equation*}
\begin{split}
\left|\Psi_{3}\right\rangle=
\frac{1}{2}
\left(\left\lvert 000000\right\rangle_{A'BDE'FH}+\lvert 111111
\rangle_{A'BDE'FH} \right)\\
\otimes\lvert 000 \rangle_{JLN}
\\
+
\frac{1}{2}
\left(
\left|000111\right\rangle_{A'BDE'FH}
+
\left|111000\right\rangle_{A'BDE'FH}\right)\\
\otimes\lvert 111 \rangle_{JLN}.
\end{split}
\end{equation*}

At step 4, $t_{1}$ and $t_{2}$ aplly C-NOT.
The state becomes
\begin{equation*}
\begin{split}
\left|\Psi_{4}\right\rangle=
\frac{1}{2}
\left(
\left\lvert 000000\right\rangle_{A'BDE'FH}+\lvert 111111
\rangle_{A'BDE'FH} \right)\\
\otimes\lvert 000 \rangle_{JLN}\\ 
+\frac{1}{2}\left(
\left|010101\right\rangle_{A'BDE'FH}
+\left|101010\right\rangle_{A'BDE'FH}
\right)\\
\otimes\left|111\right\rangle_{JLN}.
\end{split}
\end{equation*}

At step 5, $t_{1}$ and $t_{2}$ delete redundant registers $N$ and $F$
using Removal. From Lemma~\ref{lemma_rem}, the state becomes
\begin{equation*}
\begin{split}
\left|\Psi_{5}\right\rangle=
\frac{1}{2}
\left(
\left\lvert 000000\right\rangle_{A'BDE'FH}+\lvert 111111
\rangle_{A'BDE'FH} \right)\lvert 0 \rangle_{J}\\ 
+\frac{1}{2}\left(
\left|010101\right\rangle_{A'BDE'FH}
+\left|101010\right\rangle_{A'BDE'FH}
\right)
\left|1\right\rangle_{J}.
\end{split}
\end{equation*}

At step 6, $r_{2}$ deletes the redundant  register $J$ using
Removal:Add. From Lemma~\ref{lemma_remadd}, the state becomes
\begin{equation*}
\begin{split}
\left|\Psi_{6}\right\rangle=
\frac{1}{2}
\left(
\left\lvert 000000\right\rangle_{A'BDE'FH}+\lvert 111111
\rangle_{A'BDE'FH} \right)\\ 
+\frac{1}{2}\left(
\left|010101\right\rangle_{A'BDE'FH}
+\left|101010\right\rangle_{A'BDE'FH}
\right).
\end{split}
\end{equation*}

At step 7, $r_{1}$ deletes redundant registers $D$ and $H$ using
Removal the same way as in step 5. The state becomes
\begin{equation*}
\begin{split}
\left|\Psi_{7}\right\rangle=
\frac{1}{2}\left(
\left|0000\right\rangle_{A'BE'F}
+
\left|1111\right\rangle_{A'BE'F}\right)\\
+
\frac{1}{2}\left(
\left|0110\right\rangle_{A'BE'F}
+
\left|1001\right\rangle_{A'BE'F}
\right)\\
=
\lvert \Psi^{+} \rangle_{A'F}
\otimes
\lvert \Psi^{+} \rangle_{B'E}
\end{split}
\end{equation*}
We obtain separated EPR-pairs.
The first one is owned by $(s_{1},t_{1})$, and the second by $(s_{2},t_{2})$.

In the next subsection, we describe an encoding protocol without
additional registers based on the above protocol.

\subsection{Encoding without additional registers}
We now show how to use the result of the previous subsection to construct
a network coding scheme over the network of Fig.~\ref{quantumgraph}
(i.e., without additional registers).
When there are no additional registers, the input state is
\begin{equation*}
\begin{split}
\lvert \Psi_{0'} \rangle=
\lvert\Psi^{+}\rangle_{AB}\lvert\Psi^{+}\rangle_{CD}\lvert\Psi^{+}\rangle_{EF}
\lvert\Psi^{+}\rangle_{GH}\lvert\Psi^{+}\rangle_{IJ}
\lvert\Psi^{+}\rangle_{KL}\\
\otimes\lvert\Psi^{+}\rangle_{MN}.
\end{split}
\end{equation*}
Suppose that $s_{1}$ and $t_{2}$ apply ${\bf Con}^{A}_{C->D}$, and
$s_{2}$ and $t_{1}$ apply ${\bf Con}^{E}_{G->H}$. The state becomes
\begin{equation*}
\lvert \Psi_{1'} \rangle=
\lvert GHZ \rangle_{ABD}
\lvert GHZ \rangle_{EFH}
\lvert \Psi^{+} \rangle_{IJ}
\lvert \Psi^{+} \rangle_{KL}
\lvert \Psi^{+} \rangle_{MN}.
\end{equation*}
Compare with state $(\ref{inputstate})$,
the two states are the same if we take $A = A'$ and $E=E'$. Then, if
we apply steps 2-7 as in the previous section, we obtain the state
\begin{equation*}
\left|\Psi_{7'}\right\rangle=\lvert \Psi^{+} \rangle_{AF}
\otimes\lvert \Psi^{+} \rangle_{BE}.
\end{equation*}
Thus, the only modification we have to make is to replace step 1 in
the procedure of the previous subsection.
The whole procedure for network coding over the network of Fig.~\ref{quantumgraph}
is described as Table~\ref{Proc_Encwithoutvr}.
\begin{table}
\caption{\label{Proc_Encwithoutvr} Encoding without additional registers}
\begin{ruledtabular}
 \begin{tabular}{ll}
  {\bf Step 1.}& $s_{1}$ and $r_{1}$ apply ${\bf
    Con}^{A}_{C->D}$,\\
& $s_{2}$ and $r_{1}$ apply ${\bf Con}^{E}_{G->H}$.  \\
  {\bf Step 2.}& $r_{1}$ and $r_{2}$ apply ${\bf Add}^{{D,H}}_{I->J}$. \\
  {\bf Step 3.}& $r_{2}$, $t_{1}$ and $t_{2}$ apply ${\bf Fanout}^{J}_{K->L, M->N}$. \\
  {\bf Step 4.}& $t_{1}$ applies ${\bf CNOT}^{(N,F)}$, $t_{2}$ applies ${\bf CNOT}^{(L,B)}$.\\
  {\bf Step 5.}& $t_{2}$ and $r_{2}$ apply ${\bf Rem}_{L->J}$,\\
& $t_{1}$ and $r_{2}$ apply ${\bf Rem}_{N->J}$. \\
  {\bf Step 6.}& $r_{2}$ and $r_{1}$ apply ${\bf RemAdd}_{J->{D,H}}$. \\
  {\bf Step 7.}& $r_{1}$ and $s_{1}$ apply ${\bf Rem}_{D->A}$,\\
& $r_{1}$ and $s_{2}$ apply ${\bf Rem}_{H->E}$. \\
 \end{tabular}
\end{ruledtabular}
\end{table}
We show the corresponding circuit in Fig.~\ref{circuit}.
\begin{figure*}
\includegraphics[width=510pt]{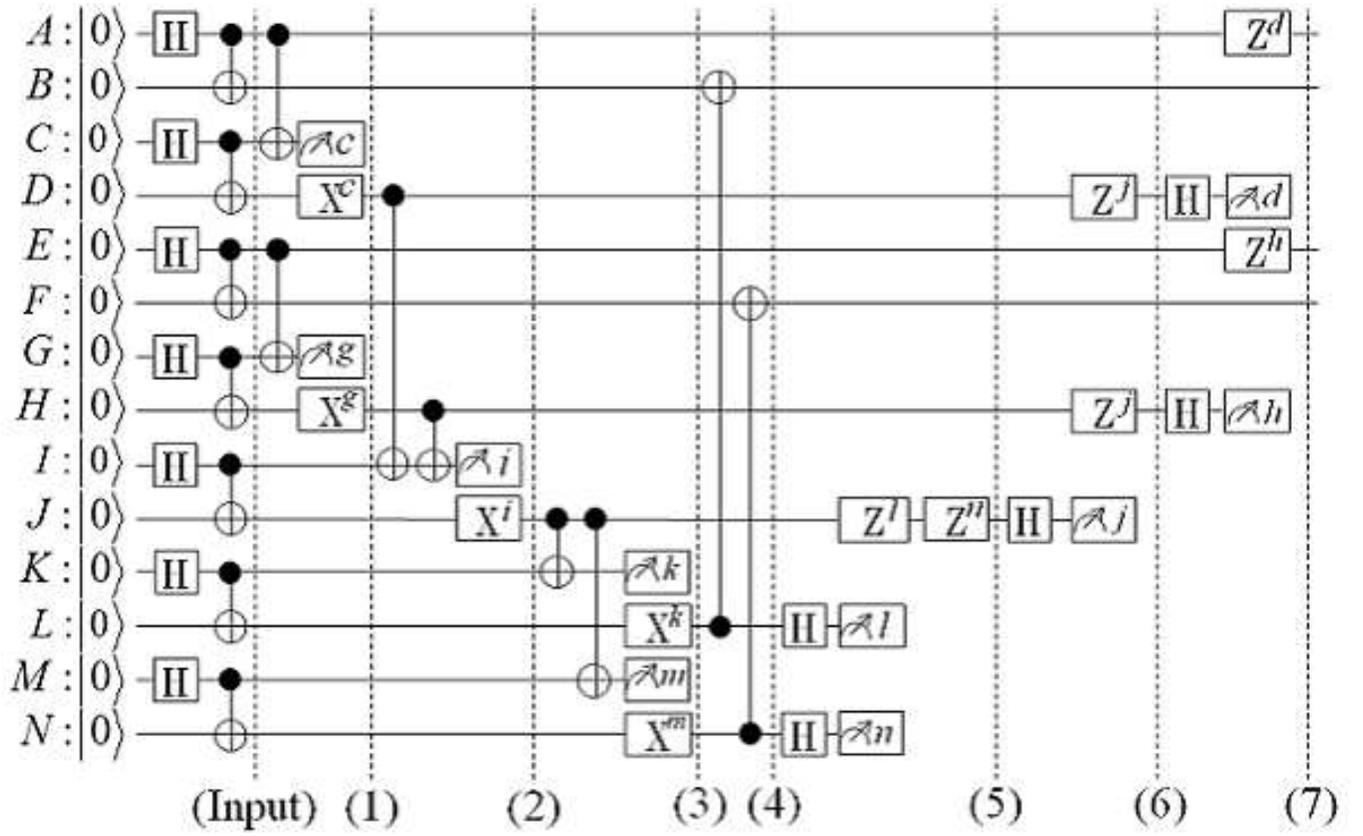}
\caption{\label{circuit}Overall view of the encoding
  circuit. Parenthetic numbers
  refer to the state after each step of the encoding procedure of
  Table~\ref{Proc_Encwithoutvr}.}
\end{figure*}

\section{Conclusion}
Our protocol shows that quantum network coding techniques are
operational using only LOCC and shared EPR-pairs between adjacent
repeaters (i.e., without additional quantum registers).
This method has been described for the butterfly network but can be
actually  extended to other linear network coding schemes on other
graphs.
We expect that this protocol will be a fundamental tool to apply
techniques from network coding to real
quantum repeater networks.
\bibliography{201107_PRA}% Produces the bibliography via BibTeX.

\end{document}